\documentclass[twoside,leqno,twocolumn]{article}
\usepackage{ltexpprt}
\usepackage{etoolbox}
\usepackage{lmodern}
\usepackage{amssymb}
\usepackage{amsmath}
\usepackage{forloop}

\usepackage[ruled,vlined,linesnumbered,norelsize]{algorithm2e}

\usepackage[hypertexnames = false]{hyperref}

\newtheorem{observation}{Observation}[section]
\newtheorem{reduction}{Reduction}

\makeatletter
\newcommand{\newreptheorem}[2]{%
	\newenvironment{rep#1}[1]{%
	\expandafter\renewcommand\csname the#1\endcsname{\ref{##1}}%
	\begin{#1}}%
	{\end{#1}%
	\addtocounter{#2}{-1}}}
\makeatother

\newreptheorem{theorem}{theorem}
\newreptheorem{reduction}{reduction}
\newreptheorem{lemma}{lemma}

\newcommand{\defcal}[1]{\expandafter\newcommand\csname c#1\endcsname{{\mathcal{#1}}}}
\newcommand{\defbb}[1]{\expandafter\newcommand\csname b#1\endcsname{{\mathbb{#1}}}}
\newcounter{calBbCounter}
\forLoop{1}{26}{calBbCounter}{
    \edef\letter{\Alph{calBbCounter}}
    \expandafter\defcal\letter
		\expandafter\defbb\letter
}

\newcommand{\ie}{{\it i.e.}}
\newcommand{\eg}{{\it e.g.}}
\newcommand{\OMAP}{{\textsf{OMAPD}}}

\begin{document}

\title{\Large  $O(\mbox{depth})$-Competitive Algorithm for Online Multi-level Aggregation}
\author
{
Niv	Buchbinder\thanks{Department of Statistics and Operations Research, School of Mathematical Sciences, Tel Aviv university, Israel. Email: \texttt{niv.buchbinder@gmail.com}. Research is supported by ISF grant 1585/15 and US-Israel BSF grant 2014414.}
\and
Moran Feldman\thanks{Department of Mathematics and Computer Science, The Open University of Israel.
Email: \texttt{moranfe@openu.ac.il}. Research supported in part by ISF grant 1357/16.}
\and
Joseph (Seffi)	Naor\thanks{Computer Science Department, Technion, Israel. Email: \texttt{naor@cs.technion.ac.il}. Research is supported by ISF grant 1585/15 and US-Israel BSF grant 2014414.}
\and
Ohad Talmon\thanks{Computer Science Department, Technion, Israel. Email: \texttt{ohad@cs.technion.ac.il}. Research is supported by ISF grant 1585/15 and US-Israel BSF grant 2014414.}
}
\maketitle

\begin{abstract}\small\baselineskip=9pt
We consider a multi-level aggregation problem in a weighted rooted tree, studied recently by Bienkowski et al.~\cite{BBBCDF15}. In this problem requests arrive over time at the nodes of the tree, and each request specifies a deadline. A request is served by sending it to the root before its deadline at a cost equal to the weight of the path from the node in which it resides to the root. However, requests from different nodes can be aggregated, and served together, so as to save on cost. The cost of serving an aggregated set of requests is equal to the weight of the subtree spanning the nodes in which the requests reside. Thus, the problem is to find a competitive online aggregation algorithm that minimizes the total cost of the aggregated requests. This problem arises naturally in many scenarios, including multicasting, supply-chain management and sensor networks. It is also related to the well studied TCP-acknowledgement problem and the online joint replenishment problem.

We present an online $O(D)$-competitive algorithm for the problem, where $D$ is the depth, or number of levels, of the aggregation tree. This result improves upon the $D^2 2^D$-competitive algorithm obtained recently by Bienkowski et al.~\cite{BBBCDF15}.

\medskip
\noindent \textbf{Keywords:} online algorithms, competitive analysis, aggregation of requests
\end{abstract}

%\thispagestyle{empty}
%\pagenumbering{Alph}
%\newpage
%\setcounter{page}{1}

%\pagenumbering{arabic}

\section{Introduction}

Aggregation of tasks is a fundamental tool in optimization, utilized in various areas. Suppose that there is a
set of requests that need to be served, where serving several requests together costs less than serving each request individually.
There are {\em aggregation constraints} on the request system that specify which sets of requests can be jointly served, together with a function that determines the cost of serving any aggregated set of requests.
For example, serving all requests together might be impossible or very costly. In general, various cost functions and aggregation constraints give rise to a large family of interesting problems. %that we are interested in.

%This happens, for example, when there is a resource needed for serving the requests which has a (large) setup cost, and its cost is independent of the number of requests using it.
Scenarios where aggregation is beneficial are common in supply-chain management, in which producing (or delivering) several demands together can be cheaper than producing them separately \cite{CW73,K97,AJR89,WW04}. Another example is aggregating control packets in communication networks (\eg, sensor networks)~\cite{DGS98,BS00,BC01,KKR01}. In this setting, transmitting together several packets going to a joint destination is, again, cheaper than transmitting each packet separately. Aggregation problems were studied extensively for many settings, both in the offline case, in which all requests are known a-priori, and in the online case in which requests arrive over time.

We study a very general online aggregation setting, considered recently by~\cite{BBBCDF15}, and known as the
\textsf{Online Multi-level Aggregation Problem with Deadlines}. In this problem requests arrive over time, and each request specifies a deadline for serving it.
Thus, each request is associated with a time interval in which it needs to be served. At any point of time $t$, the online algorithm is only aware of requests whose arrival time
is no later than $t$.

There is a tree with non-negative edge costs rooted at a node $r$, and each request is assumed to reside at some tree node.
The cost of serving a single request is equal to the cost of the tree path from the node where it resides to the root $r$.
%$[a,d]$, where $a$ is its arrival time and $d$ is its deadline.
Requests whose intervals intersect in time can be aggregated and served together. %The common resource In this problem defined via a rooted tree with non-negative edge costs.
The aggregate cost of a set of requests is equal to the cost of the subtree rooted at $r$
and spanning the nodes where these requests reside. The goal is to minimize the total {\em service cost}, \ie, the sum of the costs of the subtrees used for serving the requests.

The online multi-level aggregation problem with deadlines was studied recently by Bienkowski et al.~\cite{BBBCDF15} who designed a $D^2 2^D$-competitive algorithm for it, where $D$ denotes the depth of the tree. We note that this bound is quite far from the best known lower bound on the competitive factor for the problem which is only $2$~\cite{BBCJNS14} (see Section~\ref{ssc:related_work} for more details). In the offline case, the problem is known to be NP-hard (and APX-hard) even when the aggregation tree is of depth $2$~\cite{AJR89,NS09,BBCDN15}. Currently, the best (offline) approximation factor known for the problem is 2~\cite{BMVKSS09,BBBCDF15}.

\subsection{Our Results.}
Our main result is a competitive algorithm for the \textsf{Online Multi-level Aggregation Problem with Deadlines}.
We have already mentioned that there is a large gap between the lower bound of $2$ on the competitive ratio known and the current best upper bound of $D^2 2^D$ \cite{BBBCDF15}, which is exponential in the depth $D$ of the aggregation tree. We obtain a substantial improvement over the upper bound of~\cite{BBBCDF15}.%, proving the following.

\begin{theorem}\label{thm:node_weighted_general}
There exists an $O(D)$-competitive algorithm for the \textsf{Online Multi-level Aggregation Problem with Deadlines}.
\end{theorem}

\paragraph{Our Techniques.} Standard reductions allow us to assume that the costs are associated with the tree nodes rather than with the edges. To simplify the presentation, we first present our result for a special kind of aggregation tree known as \emph{3-decreasing} trees. In a $3$-decreasing tree, as we traverse any path from the root $r$ to a leaf, the edge costs go down by a factor of at least $3$ at each step of the path. The idea of our algorithm for $3$-decreasing trees is simple and intuitive. Suppose there is a request that has reached its deadline (and, thus, must be served). The algorithm recursively aggregates requests into a subtree (which eventually will be served), starting from the root of the tree.
%request that has just reached its deadline.
The aggregation process has to balance between two conflicting objectives. On the one hand, it is beneficial to aggregate as many requests as possible, especially those approaching their deadline. On the other hand, it is necessary to bound the cost of the aggregate subtree, so that it can be related to the optimal solution. We balance between the two by recursively providing budgets to tree nodes that are added to the aggregate subtree. A budget provided to node $u$ (that has joined the subtree) can be used to aggregate an additional (yet restricted) set of nodes into the subtree rooted at $u$. These nodes are then also given budgets, which can be used to aggregate additional nodes (recursively).

The analysis of the algorithm is performed by first showing that the total cost of each served subtree is at most $O(D \cdot c(r))$, where $D$ is the depth of the tree and $c(r)$ is the cost of the root $r$. Then, we show that the optimal cost of serving the remaining requests (including requests that at this point of time have not arrived yet) decreases by at least $c(r)$. The two claims together imply the desired bound.

Extending our result to general trees is done via a reduction that transforms a tree into a forest of $3$-decreasing trees. Then, we apply the algorithm for $3$-decreasing trees to each tree of the resulting forest. A very simple argument shows that the algorithm obtained this way has a competitive factor of at most $O(D^2)$. However, using a more involved analysis we are able to show that the two above claims hold also for general trees, which give us the promised $O(D)$-competitiveness guarantee.

\subsection{Related Work.} \label{ssc:related_work}

Several interesting problems are related to the \textsf{Online Multi-level Aggregation Problem with Deadlines}.
 %come in several interesting. In a generalization of the problem known as \textsf{Online Multi-level Aggregation Problem}
Suppose that each request accrues a {\em waiting cost} over time, rather than having a strict deadline. The cost of satisfying a set of requests is then defined to be the sum of both the service cost and the sum of the waiting costs of the requests.
For this problem Bienkowski et al.~\cite{BBBCDF15} designed an $O(D^4 2^D)$-competitive algorithm. The best offline approximation for this variant is $2+\epsilon$~\cite{P16} adapting ideas from \cite{LRS06}. We note that a common choice for the waiting cost is a linear function over time.
This model (with slight variations) was considered by Khanna et al. \cite{KNR02} who achieved a competitive factor which is logarithmic in the
weight of the aggregation tree.

One of the earliest problems considered in the general setting of aggregation is the {\em TCP-Acknowledgment Problem}. In this
problem there is a single link over which packets need to be acknowledged by control messages. Any number of control messages
can be aggregated into a single packet sent over the single link. The aggregation tree, thus, consists of a single edge. The objective function is composed of two terms: the number of acknowledgments sent over the link and the sum of the waiting times of the control messages.
The optimal competitive factors for the TCP-acknowledgment problem are $2$ for deterministic algorithms and $e/(e-1)$ if randomization is allowed~\cite{DGS98,KKR01,BJN07}. Interestingly, the TCP-acknowledgment problem is equivalent to the classical {\em Lot Sizing Problem}~\cite{DGS98,WW04},
which has been studied extensively by the operations research community.

The aggregation problem for trees of depth $2$ is known as the {\em Joint Replenishment Problem} in supply-chain management. The best competitive factor known for this special case is $3$ with general waiting times, and $2$ if requests have deadlines~\cite{BKV12,BKLMS13,BBCJNS14}. The best lower bound for general waiting times is 2.754 \cite{BBCJNS14} improving upon an earlier bound of 2.64 by \cite{BKLMS13}. For the deadline case the best lower bound is $2$ \cite{BBCJNS14}. The (offline) Joint Replenishment problem is known to be NP-hard (and APX-hard)~\cite{AJR89,NS09,BBCDN15}.
The best approximation ratio currently known for it is $1.791$~\cite{BBCJNS14}, which improves upon previous results of \cite{LRS05,LS06,LRS06,LRSS08}.

Finally, when the aggregation tree is an infinite half line, Bienkowski et al.~\cite{BBBCDF15} gave a $4$-competitive algorithm for the deadline version of the problem, and showed that this is the best possible. For the more general problem with waiting costs, Bienkowski et al.~\cite{BCJSS13} showed that the competitive ratio is between $2+\phi\approx 3.618$ and $5$, improving on an earlier upper bound of $8$ for the problem by Brito et al.~\cite{BKV12}. It is also known that the optimal offline solution for this case can be computed efficiently~\cite{BCJSS13}.

%******************************
%
%As stated the current best offline approximation for \textsf{Online Multi-level Aggregation Problem with Deadlines} is 2 \cite{BMVKSS09}. For the more general problem with waiting costs the best bound is $2+\epsilon$  \cite{28, LRS06}  **** unclear if we can use [28].
%
%P16
%Some motivating examples: [8,3,17,32,27,12,20,21]

%The TCP Acknowledgement problem is one of the earliest examples ...
%The JRP problem is yet another example ...
%
%Description of problem on multi-level trees ... weights on edges or vertices ...
%either deadline or waiting time or more complicated function.
%
%Defined in full generality by Bienkowski et al for different wait functions
%More general when the cost function is submodular
%
%Online and offline, comp. factor and approximation algorithms
%Few sample citations of papers on the problem

\section{Problem Definition and Preliminaries}\label{sec:pre}

An instance of the \textsf{Online Multi-level Aggregation Problem with Deadlines} (\OMAP) is defined as a tuple $(\cT, \cI)$. The first component of the tuple is a tree $\cT$ rooted at some node $r(\cT)$ with non-negative edge costs. We denote by $c(e)$ the cost of an edge $e$ of $\cT$, and by $D(\cT)$ the \emph{depth} of $\cT$, \ie, the maximum number of edges along any path from the root of $\cT$ to a leaf. The second component of the tuple $(\cT, \cI)$ is a set of time intervals $\cI$, where each interval $I = [a_I, d_I] \in \cI$ is associated with a node $w_I \in \cT$, an arrival time $a_I$ and a deadline $d_I$. A solution for the problem is a sequence of subtrees $T_1, T_2, \ldots, T_{\ell}\subseteq \cT$ rooted at $r(\cT)$ that are transmitted at times $t_1, t_2, \ldots, t_{\ell}$.\footnote{Notice that we abuse notation here and treat trees as sets of nodes. This is done occasionally throughout the paper.} The solution is feasible if, for each interval $I\in \cI$, there exists a tree $T_i$ with transmission time $t_i\in [a_I, d_I]$ containing the node $w_I$. Moreover, we say that a tree of the solution having these properties \emph{services} $I$. The cost of transmitting a tree $T$ is denoted by $c(T) \triangleq \sum_{e\in E[T_i]}c(e)$, where $E[T_i]$ is the set of edges of $T_i$. The cost of the solution is $\sum_{i=1}^{\ell}c(T_i)$.

In the online setting the tree $\cT$ is known to the algorithm in advance, but the intervals of $\cI$ arrive in an online fashion, \ie, each interval $I \in \cI$ is revealed to the online algorithm only upon its arrival time at $a_I$. The online algorithm can decide at every given time to transmit a subtree of $\cT$ at this time, however, both the decision to transmit the subtree and the choice of the subtree to transmit must be done without any knowledge about future intervals. The algorithm is $\alpha$-competitive if it produces a feasible solution whose total cost is always at most $\alpha$ times the cost of the optimal solution. In the online setting it is often useful to call intervals that have already been revealed to the algorithm, but were not serviced yet, by the term \emph{active}.

To simplify the presentation of our algorithm, we first modify the problem as described by the following two modifications.

\begin{enumerate}
\item We assume the tree $\cT$ has node costs $\{c(u) \mid u \in \cT\}$ rather than edge costs. Consequently, the transmission cost of a subtree $T \subseteq \cT$ now becomes $c(T) = \sum_{u\in T} c(u)$.
\item For each node $u$, we assume $c(u)$ is strictly positive.
\end{enumerate}

Proving that these modifications are without loss of generality can be done in a rather standard way. First, we may assume that $r(\cT)$ has only a single child. One may observe that the transmission of a subtree $T$ containing multiple children of $r(\cT)$ can always be replaced with the transmission of multiple subtrees, each containing $r(\cT)$ and the part of $T$ descending from a single child of $r(\cT)$. Moreover, this replacement does not change the cost of the transmission. Hence, there is no loss in applying any algorithm independently to every subtree of $\cT$ (more accurately, every instance of the algorithm faces an instance of {\OMAP} with a tree containing $r(\cT$) and all the nodes descending in $\cT$ from a single child of $r(\cT)$). Next, we observe that in {\OMAP} every interval associated with the root $r(\cT)$ can always be serviced without any cost. Thus, we may assume that our instance does not contain any such intervals. We now transform $\cT$ into a node-cost tree by moving the cost of every edge to its end point which is further away from the root, and then removing the root $r(\cT)$ itself (which is the only node left with no cost). Finally, if some node $u$ has cost zero, we may merge it with the next node on the path to the root. Thus, each node has strictly positive cost.
It is easy to see that any solution for the original edge weighted instance can be transformed to a solution for the resulting node weighted instance with the same cost and vice-versa. Moreover, the depth of the node weighted instance is always smaller than the depth of the original edge weighted instance. From now on we abuse notation and refer to the variant of {\OMAP} obtained by our modifications simply as {\OMAP}.

Finally, a tree $\cT$ is called \emph{3-decreasing} if, as we go along any path from the root $r(\cT)$ to a leaf, the cost of each node is smaller than the cost of the previous node by a factor of at least $3$. This definition is used in~\cite{BBBCDF15}, and it is similar to the weighted version of HSTs defined in~\cite{BBMN15}. Note, however, that a $3$-decreasing tree is not exactly an (un-weighted) $3$-HST (as defined by~\cite{Ba96,Ba98,FRT04}) as the costs of the edges along any root to leaf path of a $3$-decreasing tree decrease by factor of \emph{at least} $3$ rather than by exactly $3$. Moreover, different root to leaf paths of a $3$-decreasing tree might differ in their lengths.

\section{Algorithm for \texorpdfstring{$3$}{3}-Decreasing Trees} \label{sec:algorithm_decreasing}

%\SetKwFor{At}{at}{do}{end}
\begin{algorithm*}[t!]
\caption{\textsf{Transmission Tree Section}} \label{alg:online}
\DontPrintSemicolon
	%Let $r$ be the root of a $3$-decreasing tree $\cT$ containing a mature interval (that has reached its deadline). \\
    %The output is a transmission tree $T$ rooted at $r$.
    Initiate a transmission tree $T \gets \{r\}$.\\
    Let $Q$ be a queue of tuples to be processed. Initially, $Q$ contains only the tuple $(r, \hat{c}(r)=2c(r))$.\\
    \While{$Q\neq \varnothing$}
    {
    Dequeue the next tuple $(u, \hat{c}(u))$ from $Q$, and set $\cA_u \gets \varnothing$.\\
	\While{there are active intervals associated with nodes in $\cT_u \setminus T$ and $c(\cA_u) \leq \hat{c}(u) / 2$ \label{line:u_loop}}
	{
        Let $I$ be the interval with the earliest deadline among the active intervals in $\cT_u \setminus T$ (breaking ties arbitrarily).\label{line:interval_selection}\\
        Add to $\cA_u$ all nodes on the path from $w_I$ to $r$ that are not already in $T$.\\
     }
     For every $v\in \cA_u$, enqueue into $Q$ the tuple $\left(v, c(v) \cdot \frac{\hat{c}(u)}{c(\cA_u)}\right)$.\\ %\tcc*[r]{If $c(\cA_u) = 0$, then we assume $\frac{\hat{c}(u)}{c(\cA_u)}$ evaluates to $0$.}
	 Add the nodes of $\cA_u$ to $T$.	
}
Transmit $T$.
\end{algorithm*}

In this section we present and analyze an algorithm for {\OMAP} on $3$-decreasing trees. A generalization of this algorithm for general trees is given in Section~\ref{sec:algorithm_general}. Throughout the section we consider a fixed instance $(\cT, \cI)$ of {\OMAP} in which the tree $\cT$ is a $3$-decreasing tree. %The output of the algorithm is a sequence of transmission subtrees of $\cT$.
Since we have fixed the instance of {\OMAP}, we can use in this section a somewhat simplified notation. More specifically, we drop the parameter from the notations $r(\cT)$ and $D(\cT)$ and use $r$ and $D$ to denote the root of $\cT$ and its depth, respectively. % Throughout the section we assume $(\cT, \cI)$ is an instance of {\OMAP} with a $3$-decreasing tree $\cT$.
We also need to define some new notation. Given a set $U$ of nodes we use $c(U)$ to denote the total cost of the nodes within it, \ie, $c(U) = \sum_{u \in U} c(u)$. Additionally, given a node $u$ we denote by $\cT_u$ the subtree of $\cT$ rooted at $u$.

We remind the reader that an interval is said to be active if and only if it has already appeared, but has not been serviced yet. Our algorithm transmits a tree whenever an active interval {\em matures}, \ie, reaches its deadline. More specifically, our algorithm invokes Algorithm~\ref{alg:online} whenever an active interval matures, and Algorithm~\ref{alg:online} then selects a subtree $T$ of $\cT$ and transmits it. Following the transmission of $T$, our algorithm returns to its idle state until another active interval matures. We note that there might be multiple active intervals that reach maturity at the same time. When this happens, the transmission of $T$ might not service all of these intervals, which might result in immediate additional invocation of Algorithm~\ref{alg:online}. In other words, there might be a zero time gap between consecutive executions of Algorithm~\ref{alg:online}.

Informally, Algorithm~\ref{alg:online} starts the construction of the tree to be transmitted by assigning a budget of $\hat{c}(r) = 2c(r)$ to the root of the tree. This budget is then used to recursively add new nodes to the transmission tree (and thus, serve the intervals residing in these nodes).
More specifically, each node $u$ that is assigned a budget uses it to add to the transmission tree new nodes belonging to its subtree $\cT_u$. The total transmission cost of these newly added nodes is roughly equal to the budget of $u$; and following the addition of these new nodes $u$ splits its budget proportionally among them  so that this process can be repeated recursively. It is important to note that, when a node $u$ chooses nodes from its subtree $\cT_u$ to add to the transmission tree, it gives priority to satisfying intervals whose deadline is sooner (and thus are, intuitively, more ``urgent'').

The following observation shows that each transmission of Algorithm~\ref{alg:online} makes progress, and thus, our algorithm, as a whole, terminates.

\begin{observation}
Algorithm~\ref{alg:online} halts and the transmission tree $T$ contains an active interval.
\end{observation}
\begin{proof}
Observe that Line~\ref{line:interval_selection} of Algorithm~\ref{alg:online} is guaranteed to be reached at least once in every given execution of Algorithm~\ref{alg:online}. This line selects an active interval $I$, and the nodes of the path from $r$ to $w_I$ are later added to $T$.
\end{proof}

We are now ready to present the two main lemmata necessary for the analysis of the algorithm. The first of these lemmata bounds the cost paid by the algorithm for every single transmission. The proof of this lemma is deferred to Section~\ref{ssc:alg-pay}.

\begin{lemma}\label{lem:alg-pay}
For every subtree $T$ of $\cT$ transmitted by the algorithm, $c(T) \leq 2(D+1) \cdot c(r)$.
\end{lemma}

The presentation of the other main lemma requires some additional notation. Assume the online algorithm transmits overall $\ell$ trees $T_1, \ldots, T_{\ell}$. For every $0 \leq i \leq \ell$, let $\cI_{i}$ be the set of intervals that were not yet serviced by the algorithm after it has made $i$ transmissions. Note that $\cI_{i}$ includes also intervals that arrive after the transmission time of tree $T_i$, \ie, they were not active yet when $T_i$ was transmitted.
By definition, $T_i$ services all the intervals $\cI_{i-1} \setminus \cI_{i}$. Additionally, we use $OPT(\cI')$ to denote an optimal solution for a set of intervals $\cI' \subseteq \cI$ (more formally, $OPT(\cI')$ is an optimal solution for the instance $(\cT, \cI')$ of {\OMAP}).

Our second main lemma can now be stated as follows. We defer its proof to Section~\ref{ssc:opt-pay}.

\begin{lemma}\label{lem:opt-pay}
$OPT(\cI_{i}) \leq OPT(\cI_{i - 1}) -c(r)$ for every $1 \leq i \leq \ell$.
\end{lemma}

Analyzing our algorithm is now straightforward.

\begin{theorem}\label{thm:node_weighted_competitive}
There exists an $O(D)$-competitive algorithm for \textsf{Online Multi-level Aggregation Problem with Deadlines} on $3$-decreasing trees.
\end{theorem}

\begin{proof}
By Lemma~\ref{lem:alg-pay} the total cost suffered by our algorithm is at most $\ell \cdot [2(D + 1) \cdot c(r)]$. On the other hand, Lemma~\ref{lem:opt-pay} shows that the cost of the optimal solution $OPT(\cI_0) = OPT(\cI)$ is at least $\ell \cdot c(r) + OPT(\cI_\ell) = \ell \cdot c(r)$. %We complete the proof by recalling  that the cost of an optimal solution  on $\cF$ (after applying Reduction~\ref{red:combined}) cannot increase compared to an optimal solution on $\cT$.
\end{proof}

%Our result for the \textsf{Online Multi-level Aggregation Problem with Deadlines} on general trees (Theorem \ref{thm:node_weighted_general}) is obtained by combining Theorem~\ref{thm:node_weighted_competitive} with Reduction~\ref{red:combined}.

\subsection{Proof of Lemma~\ref{lem:alg-pay}.} \label{ssc:alg-pay}

In this section we prove Lemma~\ref{lem:alg-pay}. Let us first recall the lemma.

\begin{replemma}{lem:alg-pay}
%\begin{lemma}
For every subtree $T$ of $\cT$ transmitted by the algorithm, $c(T) \leq 2(D+1) \cdot c(r)$.
%\end{lemma}
\end{replemma}

%The subtree $T$ produced by Algorithm~\ref{alg:online} (from which $T'$ is generated) is obtained from a single execution of the algorithm. We first prove some claims about such an execution. %

We begin the proof of the above lemma with the following claim.

\begin{lemma} \label{lem:A_u_bound}
For every node $u$ added to $T$, $c(\cA_u) \leq \frac{1}{2}\left(\hat{c}(u)+c(u)\right)$.
\end{lemma}
\begin{proof}
Each time that Algorithm~\ref{alg:online} adds nodes to $\cA_u$, it adds only nodes belonging to the path from $u$ to $w_I$ of a single interval $I$. Since $T$ is $3$-decreasing, the total cost of the nodes on such a path is upper bounded by
\[
	\sum_{i = 1}^D \left(\frac{1}{3}\right)^i \cdot c(u)
	\leq
	\sum_{i = 1}^\infty \left(\frac{1}{3}\right)^i \cdot c(u)
	=
	\frac{c(u)/3}{1 - 1/3}
	=
	\frac{c(u)}{2}
	\enspace.
\]
On the other hand, Algorithm~\ref{alg:online} stops adding nodes to $\cA_u$ once $c(\cA_u)$ exceeds $\frac{\hat{c}(u)}{2}$, thus completing the proof.
\end{proof}

Next, define inductively the level of nodes in $T$ as follows. The level of the root $r$ is $0$. For node $v$, suppose that node $u$ added it to $\cA_u$; then the level of $v$  is defined to be the level of $u$ plus $1$. Since the nodes of $\cA_u$ are all descendants of $u$, the level of each node $u\in \cT$ is guaranteed to be at most $D$.

\begin{lemma} \label{lem:c_upper_bound}
For every node $u \in T$, $c(u) \leq \hat{c}(u)$.
\end{lemma}
\begin{proof}
We prove the lemma by induction on the level of $u$. For the root $r$ the lemma is immediate since $\hat{c}(u) = 2c(u)$. Next, let $v$ be a node in $\cA_u$ (added by $u$). Then,
\begin{align*}
	\hat{c}(v)
	& =
	c(v) \cdot \frac{\hat{c}(u)}{c(\cA_u)} \\
	& \geq
	c(v) \cdot \frac{\hat{c}(u)}{\frac{1}{2}\left(\hat{c}(u)+c(u)\right)} \\
	& \geq
	c(v) \cdot \frac{\hat{c}(u)}{\hat{c}(u)}
	=
	c(v),
	%\enspace,
\end{align*}
where the first inequality follows from Lemma~\ref{lem:A_u_bound}, and the second inequality follows from the induction hypothesis applied to $u$.
%
%We still need to take care of the case $c(\cA_u) = 0$. In this case Algorithm~\ref{alg:online} sets $\hat{c}(v) = 0$. However, this case also implies $c(v) = 0$ since $v \in \cA_u$, and thus, the lemma holds in this case as well.
\end{proof}

Note that for every node $u$ such that $\cA_u \neq \varnothing$, by construction,
\[
	\sum_{v \in \cA_u} \hat{c}(v)
	=
	\sum_{v \in \cA_u} \left[c(v) \cdot \frac{\hat{c}(u)}{c(\cA_u)}\right]
	=
	\hat{c}(u).
\]
 %When $\cA_u = \varnothing$ we get $\sum_{v \in \cA_u} \hat{c}(v) = 0 \leq \hat{c}(u)$.
Hence, the following observation can be proved by induction.

\begin{observation} \label{obs:level_cost}
In each level of $T$,
the sum of $\hat{c}(u)$, taken over all nodes $u$, is at most $\hat{c}(r) = 2c(r).$
\end{observation}

We now observe the following.
\[
	c(T)
	=
	\sum_{u\in T} c(u)
	\leq
	\sum_{u\in T}\hat{c}(u)
	\leq
	2(D+1) \cdot c(r)
	\enspace.
\]
The first inequality follows since Lemma~\ref{lem:c_upper_bound} guarantees that $c(u)\leq \hat{c}(u)$ for every node $u \in T$. The second inequality follows since there are at most $D+1$ possible levels, and Observation~\ref{obs:level_cost} shows that the sum of $\hat{c}(u)$ in each level is at most $2c(r)$.

\subsection{Proof of Lemma \ref{lem:opt-pay}.} \label{ssc:opt-pay}

In this section we prove Lemma~\ref{lem:opt-pay}. We begin by recalling the lemma.

\begin{replemma}{lem:opt-pay}
%\begin{lemma}
$OPT(\cI_{i}) \leq OPT(\cI_{i - 1}) -c(r)$ for every $1 \leq i \leq \ell$.
\end{replemma}
%\end{lemma}

To prove the lemma it suffices to construct a solution $S$ servicing all the intervals of $\cI_i$ whose cost is at most $OPT(\cI_{i - 1})-c(r)$. The lemma will then follow since, being an optimal solution, $OPT(\cI_{i})$ is not more expensive then any other feasible solution for $(\cT, \cI_i)$.

Let us now construct the above mentioned solution $S$. Recall that $T_i$ is the subtree of $\cT$ transmitted by the algorithm at its $i$-th transmission, and let us denote by $t_i$ the time of this transmission. There are a few assumptions that we can make about $OPT(\cI_{i - 1})$. First, we may assume that $OPT(\cI_{i - 1})$ transmits at most one subtree at every given time (otherwise, we may merge subtrees transmitted at the same time without increasing the total cost). Second, we may assume that $OPT(\cI_{i - 1})$ makes its first transmission at time $t_i$ because that is the time of the earliest deadline among the deadlines of the intervals of $\cI_{i - 1}$ (specifically, the interval whose deadline triggered the $i$-th transmission of the algorithm is in $\cI_{i - 1}$ and its deadline is $t_i$). Given these assumptions, there must be exactly one transmission of $OPT(\cI_{i - 1})$ at time $t_i$. Let us denote the subtree transmitted at this time by $T^*_i$.

We obtain the solution $S$ from the optimal solution $OPT(\cI_{i-1})$ by applying the following two steps to the last solution.
\begin{enumerate}
\item Remove $T^*_{i}$ from the solution. \label{step-1}
\item {\bf Reconstruction step:} Scan the intervals of $\cI_{i}$ in a non-decreasing order of their deadline. For every such interval $I$ we do the following to guarantee that it is serviced by $S$. \label{step-2}
\begin{itemize}
\item Find the subtree of the current solution transmitted within the range $[t_i, d_I]$ which contains the largest fraction of the path from $r$ to $w_I$ (breaking ties arbitrarily), and add to this subtree the remaining part of this path. Note that the above tree might contain all the path, in which case there is no need to add anything. %(so that now this subtree serves $I$).
\item If the current solution makes no transmissions within the range $[t_i, d_I]$, then introduce into it a new transmission transmitting the subtree consisting solely of the path from $r$ to $w_I$. The location of this new transmission within the range $[t_i, d_I]$ can be chosen arbitrarily.
\end{itemize}
\end{enumerate}

\begin{observation}
The solution $S$ is a feasible solution for $(\cT, \cI_{i})$.
\end{observation}
\begin{proof}
The construction of $S$ guarantees that every interval of $\cI_{i}$ is serviced by $S$ before its deadline.
\end{proof}

We are left to prove that the cost of $S$ is at most $OPT(\cI_{i})-c(r)$. It is clear that Step~\ref{step-1} in the procedure for constructing $S$ decreases the cost by $c(T^*_{i})$, thus, it suffices to show that the second step increases the cost of $S$ by at most
$c(T^*_{i})-c(r)$. The rest of this section is devoted to proving this claim.

We say that a node $u$ is \emph{reconstructed} whenever it is added to some transmission of the solution during Step~\ref{step-2} of the above procedure.

\begin{lemma} \label{lem:reconst_basic}
The reconstruction step obeys the following claims.
\begin{enumerate}
\item Only nodes of $T^*_{i}$ are reconstructed.
\item Each node of $T^*_{i}$ is reconstructed at most once.
\end{enumerate}
\end{lemma}

\begin{proof}
Consider a node $u$ which is reconstructed, and let $I$ be the interval whose processing by the reconstruction step caused the first reconstruction of $u$. Thus, $u$ is on the path from $r$ to $w_I$. If $u$ does not belong to $T^*_i$, then $OPT(\cI_{i - 1})$ must service $I$ by some transmission other than $T^*_i$ within the range $(t_i, d_I]$. Since this transmission is not removed, the processing of $I$ by the reconstruction step could not cause the reconstruction of any node, which contradicts the definition of $u$. Hence, any node $u$ which is reconstructed must belong to $T^*_i$.

It remains to prove the second part of the lemma. The fact that $u$ was reconstructed when $I$ was processed means that prior to $I$'s processing the interval $[t_i, d_I]$ did not contain any transmission involving $u$. Since the reconstruction step scans the intervals in a non-decreasing order of their deadlines, this means that any interval $I'$ for which $u$ is on the path from $r$ to $w_{I'}$ must have a deadline $d_{I'} \geq d_I$. Hence, when $I'$ is processed by the reconstruction step, the range $[t_i, d_{I'}] \supseteq [t_i, d_I]$ already includes a subtree containing $u$, and thus, $I'$ does not cause another reconstruction of the node $u$.
% assume towards a contradiction that some node $u$ is reconstructed twice: the first time when an interval $I$ is consider by the reconstruction step, and a second time when an interval $I'$ is considered. Thus, $d(I) \leq d(I')$. Since $u$ is constructed both when $I'$ and $I'$ are considered, $u$ must be on the path from $r$ to both $w_I$ and $w_{I'}$. Thus, after $I$ is considered the interval $[t, d(I)]$ is guaranteed to contain a subtree containing $u$. When $I'$ is considered the interval $[t, d(I')] \supseteq [t, d(I')]$ also contains this subtree, and thus, the tree along this range which contains the largest fraction of the path from $r$ to $w_{I'}$ already contains $u$. Since the reconstruction step never creates duplicate nodes in a subtree, it does add $u$ again to this subtree, which contradicts our assumption.
\end{proof}

The last lemma shows that only a limited set of nodes (namely, the nodes of $T^*_i$) might be reconstructed. Lemma~\ref{lem:goodsubset} shows that even within this limited set there is a significant subset of nodes that are not reconstructed. The following lemma proves a few technical claims used later in the proof of Lemma~\ref{lem:goodsubset}.

\begin{lemma}\label{lem:reconst}
Let $u\in T_{i}$ be a node that is reconstructed. Then, $c(\cA_u)>\frac{\hat{c}(u)}{2}$ and $\cA_u \subseteq T^*_{i}$.
\end{lemma}

\begin{proof}
Let $I$ be the interval that caused the reconstruction of $u$. Since $I$ belongs to $\cI_i$ (intervals outside of $\cI_i$ are not processed by the reconstruction step), it must be that $I$ is not serviced by $T_i$. On the other hand, the fact that $I$ caused the reconstruction of nodes means that it is not serviced by any subtree of $OPT(\cI_{i - 1})$ other than $T^*_i$, thus, it must be active at time $t_i$.

The above observations imply that when Algorithm~\ref{alg:online} constructed $T_i$ it left the inner loop growing $\cA_u$ while there were still active intervals associated with nodes of $\cT_u \setminus T$; which can only happen when $c(\cA_u)>\frac{\hat{c}(u)}{2}$. Next, consider any node $v\in \cA_u$. Since Algorithm~\ref{alg:online} scans the intervals in a non-decreasing deadline order while growing $\cA_u$, $v$ must have been added to $A_u$ due to being on the path from $r$ to $w_{I'}$ of some interval $I'$ having a deadline $d_{I'} \leq d_I$. Assume towards a contradiction that $T^*_i$ does not contain this path. Clearly, this assumption implies that $OPT(\cI_{i - 1})$ services $I'$ by some subtree transmitted during the range $(t_i, d_{I'}]$. Since any subtree servicing $I'$ must include $u$, $u$ is already present within the range $(t_i, d_{I'}] \subseteq [t_i, d_I]$ when $I$ is processed by the reconstruction step; which contradicts the definition of $I$ as the interval whose processing caused the reconstruction of $u$.
\end{proof}

\begin{lemma}\label{lem:goodsubset}
There exists a set $U$ of nodes obeying the following properties.
\begin{enumerate}
\item $U \subseteq T^*_{i}\cap T_i$. \label{prop-0}
\item $\sum_{u \in U} \hat{c}(u) = \hat{c}(r)=2c(r)$. \label{prop-2}
\item $\hat{c}(u) \leq 2c(u)$ for all $u\in U$.\label{prop-3}
\item The nodes of $U$ are not reconstructed. \label{prop-1}
\end{enumerate}
\end{lemma}

\begin{proof}
We start with a set $U$ obeying all the properties other than Property~\ref{prop-1}, and let it evolve while maintaining these properties. The evolution of $U$ ends as soon as it obeys also Property~\ref{prop-1}. More specifically, we initially set $U=\{r\}$. Observe that this set indeed satisfies all the properties other than Property~\ref{prop-1}. If $U$ also satisfies Property~\ref{prop-1} then we are done. Otherwise, there must be a node $u$ in the current set $U$ which is reconstructed. Since we maintain $U$ as a set obeying Property~\ref{prop-0}, $u$ must belong to $T_i$. Hence, by Lemma~\ref{lem:reconst}, it must hold that $c(\cA_u) \geq \frac{1}{2}\hat{c}(u)$ and $\cA_u \subseteq T^*_i$.

At this point we evolve $U$ by removing $u$ from it and adding the nodes of $\cA_u$ instead. Since $\cA_u \subseteq T_i \cap T^*_i$, Property~\ref{prop-0} is preserved. Additionally, $\sum_{v \in \cA_u} \hat{c}(v) = \sum_{v \in \cA_u} \left[c(v) \cdot \frac{\hat{c}(u)}{c(\cA_u)}\right]=\hat{c}(u)$, and thus, Property~\ref{prop-2} is preserved as well. Finally, Property~\ref{prop-3} also remains valid since Algorithm~\ref{alg:online} sets $\hat{c}(v)= c(v) \cdot \frac{\hat{c}(u)}{c(\cA_u)} \leq 2 c(v)$ for each node $v\in \cA_u$ (where the inequality holds since Lemma~\ref{lem:reconst} implies $c(\cA_u) \geq \frac{1}{2}\hat{c}(u)$).

We can now repeat the above evolution step as long as Property~\ref{prop-1} is violated. However, this evolution cannot continue forever since each step of it replaces a single node with nodes appearing lower than it in $\cT$. Hence, the evolution presented is guaranteed to end up eventually with a set $U$ obeying Property~\ref{prop-1} (as well as the three other properties).
\end{proof}

To conclude the proof of Lemma~\ref{lem:opt-pay} we note that, by combining all the properties of $U$ in Lemma~\ref{lem:goodsubset}, we get that there must exist a set $U$ of node in $T^*_i$ which are not reconstructed and have a total cost of at least $\sum_{u\in U}c(u) \geq \frac{1}{2}\sum_{u\in U}\hat{c}(u) = c(r)$. Together with the claim of Lemma~\ref{lem:reconst_basic} that only nodes of $T^*_i$ are reconstructed, and even they can only be reconstructed once, we get that the increase in the cost of $S$ during the reconstruction step is at most $c(T^*_i) - \sum_{u\in U}c(u) \leq c(T^*_i) - c(r)$.

\section{Algorithm for General Trees} \label{sec:algorithm_general}

In this section we show how the algorithm from Section~\ref{sec:algorithm_decreasing} can be modified to be $O(D(\cT))$-competitive also for {\OMAP} on general trees. We begin by describing a reduction, originally showed by~\cite{BBBCDF15}, transforming a general tree into a forest of $3$-decreasing trees.

Given a tree $\cT$, we construct a forest $\cF$ of $3$-decreasing trees as follows. The nodes of $\cF$ are the same as the nodes of $\cT$. We connect each node $u$ with an edge to the first node $v$ on the path from $u$ to $r(\cT)$ whose cost obeys $c(v) \geq 3c(u)$. If there is no such $v$, then $u$ becomes the root of a new tree.
Each node $u$ in the forest is now associated with a set $B_u$ of all nodes on the path from $u$ to $v$ in the original tree $\cT$ (without $v$, but including $u$ itself). If $v$ does not exist (\ie, $u$ is a root in the new forest), then the associated set $B_u$ is defined as the set of nodes on the path from $u$ to the root $r(\cT)$ of the original tree (this time, including $r(\cT)$).

\begin{observation}
The forest $\cF$ consists solely of $3$-decreasing trees.% of height at most $D(\cT)$.
\end{observation}
\begin{proof}
By definition, if $u$ is a node of $\cF$ which is not a root of its tree, then the father $v$ of $u$ in $\cF$ obeys $c(v) \geq 3c(u)$. %Moreover, $v$ is also an ancestor of $u$ in $\cT$, and thus, no tree of $\cF$ can be higher than $\cT$.
\end{proof}

Assume that $\cF$ consists of $m$ trees $\cT^1, \cT^2, \dotsc, \cT^m$. Notice that $\cT$ and $\cF$ have the same set of nodes, and thus, each interval is naturally associated with a node in one of the trees in the forest. Let $\cI^i$ be the set of intervals associated with the nodes of $\cT^i$.

Our algorithm for general trees runs an independent instance of the algorithm for $3$-decreasing trees from Section~\ref{sec:algorithm_decreasing} on each tree $\cT^i$ with its corresponding set of intervals $\cI^i$. For convenience, we denote the algorithm from Section~\ref{sec:algorithm_decreasing} by $ALG$ from this point on.  Whenever an instance of $ALG$ chooses to transmit a subtree $T \subseteq \cT^i$, we transmit instead the tree $\bigcup_{u \in T} B_u$. It is useful to call $T$ the \emph{virtual} tree transmitted by $ALG$, and $\bigcup_{u \in T} B_u$ the \emph{concrete} tree transmitted by $ALG$. Observe that the concrete tree is always a subtree of the original tree $\cT$, and thus, this is a description of a valid algorithm for the original problem ({\OMAP} on general trees).

We next show that the total cost of the optimal solutions for all the subtrees of the forest $\cF$ is bounded by the optimal cost of the original instance.
Formally, let $OPT(\cT', \cI')$ denote the cost of the optimal solution for an instance $(\cT', \cI')$ of {\OMAP}. Then,

\begin{observation} \label{obs:opt_partitioning}
$OPT(\cT, \cI) \geq \sum_{i = 1}^m OPT(\cT^i, \cI^i)$.
\end{observation}
\begin{proof}
Consider the optimal solution for the instance $(\cT, \cI)$ of {\OMAP}. Based on this optimal solution we construct a solution for every one of the instances $\{(\cT^i, \cI^i)\}_{i = 1}^m$ as follows. Whenever the optimal solution for $(\cT, \cI)$ transmits a subtree $T^*$, the solution for the instance $(\cT^i, \cI^i)$ transmits the subtree $T^* \cap \cT^i$. One can verify that $T^* \cap \cT^i$ is indeed a tree since the set of nodes on the path connecting every node $u \in \cT^i$ to $r(\cT^i)$ is a subset of the set of nodes connecting $u$ to $r(\cT)$ in $\cT$.

We complete the proof by observing that the costs of the above solutions add up to exactly $OPT(\cT, \cI)$, and thus, the total cost of the optimal solutions for the instances $\{(\cT^i, \cI^i)\}_{i = 1}^m$ cannot exceed this value.
\end{proof}

The above observation implies that in order to analyze our algorithm it is enough to relate the cost of the concrete trees transmitted by each instance of $ALG$ to the cost of the optimal solution for the instance of {\OMAP} faced by this instance of $ALG$. Notice that this is slightly different from what we do in Section~\ref{sec:algorithm_decreasing} since in Section~\ref{sec:algorithm_decreasing} we relate the cost of the \emph{virtual} trees transmitted by an instance of $ALG$ to the cost of the optimal solution for the instance of {\OMAP} faced by this instance of $ALG$. Nevertheless, we show in the rest of this section that the arguments from Section~\ref{sec:algorithm_decreasing} can be used, almost without change, to prove also the more ambitious goal we need to prove here. More specifically, we prove the following proposition.

\begin{proposition} \label{prp:single_instance_claim}
For every $1 \leq i \leq m$, the total cost of the concrete trees transmitted by the instance of $ALG$ assigned to $\cT^i$ is at most $O(D(\cT)) \cdot OPT(\cT^i, \cI^i)$.
\end{proposition}

Clearly Theorem~\ref{thm:node_weighted_general} follows from Observation~\ref{obs:opt_partitioning} and Proposition~\ref{prp:single_instance_claim}. To prove Proposition~\ref{prp:single_instance_claim} we need to define some additional notation. Assume $ALG$ transmit $\ell^i$ trees when given $(\cT^i, \cI^i)$. For every $0 \leq j \leq \ell^i$, let $\cI^i_j$ be the set of intervals from $\cI^i$ that were not yet serviced by the algorithm after it has made $i$ transmissions. Note that $\cI^i_j$ includes also intervals that arrive after the $i$-th transmission, \ie, they were not active yet when this transmission was made. Using this notation we can state the following lemma, which is a counterpart of Lemma~\ref{lem:opt-pay}, and also follows from it.

\begin{lemma}\label{lem:opt-pay-general}
$OPT(\cI^i_j) \leq OPT(\cI^i_{j - 1}) - c(r(\cT^i))$ for every $1 \leq i \leq m$ and $1 \leq j \leq \ell^i$.
\end{lemma}

To prove Proposition~\ref{prp:single_instance_claim} we also need the following counterpart of Lemma~\ref{lem:alg-pay}.

\begin{lemma}\label{lem:alg-pay-general}
If the instance of $ALG$ corresponding to $(\cT^i, \cI^i)$ transmits a virtual tree $T$, then $c\left(\bigcup_{u \in T} B_u\right) \leq \sum_{u \in T} c(B_u) \leq 6(D(\cT)+1) \cdot c(r(\cT^i))$.
\end{lemma}

One can verify that Lemmata~\ref{lem:opt-pay-general} and~\ref{lem:alg-pay-general} imply Proposition~\ref{prp:single_instance_claim} in the same way that Lemmata~\ref{lem:alg-pay} and~\ref{lem:opt-pay} imply Theorem~\ref{thm:node_weighted_competitive}. Thus, it remains to prove Lemma~\ref{lem:alg-pay-general}. Recall that $ALG$ generates each one of the virtual trees it transmits by executing Algorithm~\ref{alg:online}, and consider the execution of Algorithm~\ref{alg:online} which generated the virtual tree $T$.

The sets $\{B_u\}_{u \in T}$ might not be disjoint. However, for the purposes of the proof it is useful to assume they are disjoint. In other words, if a node $v$ appears in several sets out of $\{B_u\}_{u \in T}$, we treat each one of its appearances as unique. Additionally, let us use the shorthand $T' = \bigcup_{u \in T} B_u$.

By Lemma~\ref{lem:c_upper_bound}, $c(u) \leq \hat{c}(u)$ for every node $u \in T$. Recall that $\hat{c}(u)$ is defined by Algorithm~\ref{alg:online} only for nodes $u \in T$. We now extend the definition of $\hat{c}$ to all the nodes of $T'$ by setting $\hat{c}(v) = \hat{c}(u)$ for every node $v \in B_u$. Since, by definition, the nodes of $B_u$ have costs of at most $3c(u)$, we get the following corollary.

\begin{corollary} \label{cor:c_chat_general}
For every node $u \in T'$, $c(u) \leq 3\hat{c}(u)$.
\end{corollary}

Next, define levels for the nodes of $T'$ using the following recursive procedure. The level of the root $r(\cT^i)$ is $0$. Consider now a node $u \in T$ which already has a level $\ell_u$. Then, for every node $v \in \cA_u$, the nodes of $B_v$ are assigned the levels $\ell_u + 1, \ell_u + 2, \dotsc, \ell_u + |B_v|$; where the last level is assigned to $v$ itself. The definition of $\cA_u$ guarantees that all the nodes of $B_v$ appear on the path from $u$ to $v$ in $\cT$, and thus, the difference between the level of $u$ and $v$ is at most the difference between their heights in $\cT$. Thus, no node of $T'$ is given a level larger than $D(\cT)$ by the above procedure.\footnote{There is an alternative way to define this level assignment which might help to understand the intuition behind it. Consider a tree $T_\cA$ defined as follows. The root of $T_\cA$ is $r(\cT^i)$, and the children of every node $u \in T_\cA$ are the nodes of $\cA_u$. One can verify that $T_\cA$ has exactly the same set of nodes as $T$. Moreover, the height of every node in $T_\cA$ corresponds to its level according to the level assignment used in Section~\ref{sec:algorithm_decreasing}. If we now replace every node $u \in T_\cA$ with a path consisting of the nodes of $B_u$ in which $u$ is the lowest node, then we get a new tree $\cT_\cB$. This time, this tree has the same set of nodes as $T'$, and the height of every node in $T_\cB$ corresponds to its level according to the level assignment used in this section.}

Recall that for every node $u \in T$ such that $\cA_u \neq \varnothing$, by construction,
\[
	\sum_{v \in \cA_u} \hat{c}(v)
	=
	\sum_{v \in \cA_u} \left[c(v) \cdot \frac{\hat{c}(u)}{c(\cA_u)}\right]
	=
	\hat{c}(u).
\]
Let us now define for a node $u \in T$ an additional set $\cB_u = \{w \in B_v \mid v \in \cA_u, \ell_w = \ell_u + 1\}$. In other words, $\cB_u$ is obtained from $\cA_u$ by replacing every node $v \in \cA_u$ with the single node $w \in B_v$ whose level is larger by $1$ than the level of $u$. Since $\hat{c}(w)$ is identical for every node of $B_v$ we immediately get also $\sum_{v \in \cB_u} \hat{c}(v) = \hat{c}(u)$ whenever $\cB_u \neq \varnothing$. For a node $u \in T' \setminus T$ we define $\cB_u = \{w \in B_v \mid \ell_w = \ell_u + 1 \wedge \exists_{v \in T} u,w \in B_v\}$. Informally, $\cB_u$ contains the single node that belongs to the same set $B_v$ as $u$ and has a level larger by $1$ than $u$. Again, the fact that $\hat{c}(w)$ is identical for every node of a single set $B_v$ guarantees $\sum_{v \in \cB_u} \hat{c}(v) = \hat{c}(u)$.

\begin{lemma} \label{lem:level_cost_general}
In each level $0 \leq \ell \leq D(\cT)$, the sum of $\hat{c}(u)$, taken over all nodes of level $\ell$, is at most $\hat{c}(r(\cT^i)) = 2c(r(\cT^i)).$
\end{lemma}
\begin{proof}
We have seen that $\sum_{v \in \cB_u} \hat{c}(v) = \hat{c}(u)$ for every node $u \in T'$ unless $\cB_u = \varnothing$. Additionally, the definition of the sets $\cB_u$ guarantees that each node $v \in T'$ belongs to a single set $\cB_u$, and this set is associated with a node $u$ of a lower level than $v$. All this observations, taken together, imply the lemma by a standard induction argument.
\end{proof}

We now observe the following.
\begin{align*}
	c(T')
	={} &
	\sum_{u\in T'} c(u)
	\leq
	3 \cdot \sum_{u\in T'}\hat{c}(u)\\
	\leq{} &
	3 \cdot [2(D(\cT)+1) \cdot c(r(\cT^i))]\\
	={} &
	6(D(\cT) + 1) \cdot c(r(\cT^i))
	\enspace.
\end{align*}
The first inequality follows since Corollary~\ref{cor:c_chat_general} guarantees that $c(u)\leq 3\hat{c}(u)$ for every node $u \in T'$. The second inequality follows since there are at most $D(\cT)+1$ possible levels, and Lemma~\ref{lem:level_cost_general} shows that the sum of $\hat{c}(u)$ in each level is at most $2c(r(\cT^i))$.

\section{Conclusions}

In this paper we have presented an $O(D)$-competitive algorithm for the \textsf{Online Multi-level Aggregation Problem with Deadlines}. This result represents an exponential improvement over the previously best competitive ratio of $D^2 2^D$ given by~\cite{BBBCDF15}. Nevertheless, the competitive ratio of our algorithm is still quite far from the constant lower bounds proved by~\cite{BBCJNS14} and~\cite{BBBCDF15}. Narrowing this gap, either by providing an improved algorithm or by proving stronger lower bounds, is an intriguing open problem that we leave open.

%Another interesting question is to extend the results of this paper to the more general version of the problem in which each request accrues a {\em waiting cost} over time, rather than having a strict deadline. For this version of the problem the currently best competitive ratio is $O(D^4 2^D)$, which is still exponential in $D$~\cite{BBBCDF15}. We suspect that our result can be extended to this, more general, setting; however, the extension might require some loss in the competitive factor.

%\Mnote{* I would expect that the current result can also be extended to waiting times. It would be nice if the authors could comment on this (there is a range of possible situations: we are working on this / there is the following significant problem / we do not care).} 

\apptocmd{\sloppy}{\hbadness 10000\relax}{}{}
\bibliographystyle{plain}
\bibliography{Aggregation}

\end{document}